\documentclass{article}
\usepackage{ijcai23}

\usepackage{times}
\usepackage{soul}
\usepackage{url}
\usepackage[hidelinks]{hyperref}
\usepackage[utf8]{inputenc}
\usepackage[small]{caption}
\usepackage{graphicx}
\usepackage{amsmath}
\usepackage{amsthm}
\usepackage{amssymb}
\usepackage{mathtools}
\usepackage{cleveref}
\usepackage{booktabs}
\usepackage{algorithm}
\usepackage{algorithmic}
\usepackage[switch]{lineno}
\usepackage{dsfont}
\usepackage{todonotes}
\usepackage{enumitem}
\usepackage{paralist}
\usepackage{tikz}
\usetikzlibrary{positioning}
\usetikzlibrary{bending}
\usetikzlibrary{arrows.meta}
\usepackage{thmtools}

\newcommand*{\citet}[1]{\citeauthor{#1}~\shortcite{#1}}

\newif\iflong
\newif\ifshort

\longtrue

\iflong
\else
\shorttrue
\fi

\newcommand*{\N}{\mathds{N}}

\newtheorem{theorem}{Theorem}
\newtheorem{lemma}[theorem]{Lemma}
\newtheorem{corollary}[theorem]{Corollary}
\newtheorem{definition}[theorem]{Definition}

\newtheorem{proposition}[theorem]{Proposition}

\newcommand*{\reach}[1]{R_{#1}}
  
\newcommand*{\cost}[3]{c_{#3}(#1,#2)} 
\newcommand*{\bigconstant}{K} 
\DeclareMathOperator{\socialcost}{SC}
\newcommand*{\SC}[2]{\socialcost_{#2}(#1)} 
\newcommand*{\Gu}{G^{^\leftrightarrow}}
\newcommand{\lifetime}{t}
\newcommand{\maxLifetime}{t_{max}}

\newcommand*{\Pow}{\mathcal{P}}

\newcommand*{\s}{\mathbf{s}}
\def\poa {{\sf PoA}}

\def\opt {{\sf OPT}}

\allowdisplaybreaks

\newcommand\nod{}
\def\nod(#1,#2,#3,#4){\node[draw, circle,  minimum size=15pt, inner sep=0.5pt](#1)[#2=#4 of #3]{$#1$};}
\newcommand\ed{}
\def\ed(#1,#2,#3){(#1) edge node[above, inner sep=2pt,]{$#3$} (#2)}
\newcommand\edo{}
\def\edo(#1,#2,#3,#4){(#1) edge node[#4, inner sep=2pt,]{$#3$} (#2)}

\tikzstyle{nod} = [inner sep=2pt, draw, circle, minimum size=15pt]

\title{Temporal Network Creation Games}

\author{
Davide Bilò$^1$
\and
Sarel Cohen$^2$\and
Tobias Friedrich$^2$\and
Hans Gawendowicz$^2$\and\\
Nicolas Klodt $^2$\and
Pascal Lenzner$^2$\And
George Skretas$^2$
\affiliations
$^1$ University of L’Aquila, L'Aquila, Italy\\
$^2$Hasso Plattner Institute, University of Potsdam, Potsdam, Germany
\emails
$^1$davide.bilo@univaq.it,
$^2$\{firstname.lastname\}@hpi.de
}

\begin{document}

\maketitle

\begin{abstract}
Most networks are not static objects, but instead they change over time. This observation has sparked rigorous research on temporal graphs within the last years. In temporal graphs, we have a fixed set of nodes and the connections between them are only available at certain time steps. This gives rise to a plethora of algorithmic problems on such graphs, most prominently the problem of finding temporal spanners, i.e., the computation of subgraphs that guarantee all pairs reachability via temporal paths. To the best of our knowledge, only centralized approaches for the solution of this problem are known. However, many real-world networks are not shaped by a central designer but instead they emerge and evolve by the interaction of many strategic agents. This observation is the driving force of the recent intensive research on game-theoretic network formation models.      

In this work we bring together these two recent research directions: temporal graphs and game-theoretic network formation. As a first step into this new realm, we focus on a simplified setting where a complete temporal host graph is given and the agents, corresponding to its nodes, selfishly create incident edges to ensure that they can reach all other nodes via temporal paths in the created network. This yields temporal spanners as equilibria of our game. 
We prove results on the convergence to and the existence of equilibrium networks, on the complexity of finding best agent strategies, and on the quality of the equilibria. By taking these first important steps, we uncover challenging open problems that call for an in-depth exploration of the creation of temporal graphs by strategic agents.
\end{abstract}

\section{Introduction}
Networks are omnipresent in everyday life. They range from abstract constructs, such as (online) social networks, to essential infrastructure, such as transportation networks and power grids. Given their ubiquity and importance, rigorous research has been conducted to better understand real-world networks. Researchers strive to evaluate the behavior of networks, their structural properties and the processes that drive their formation. Over the years, the research community has realized that more varied and complex models are required to capture the intricacies that govern a network's attributes. Two such intricacies are:
\begin{compactitem}
    \item[(i)] many networks are dynamic in nature, i.e. the nodes and/or the connections of the nodes change over time;
    \item[(ii)] the formation of many networks is driven by many individual and selfish agents without central coordination.
\end{compactitem}
Due to the complexity that each of these settings introduces, researchers so far considered only one of the above assumptions that many real-world networks naturally exhibit.

However, in many real-world settings both (i) and (ii) apply. For example, consider the problem of scheduling meetings in a large institution where employees are interested in disseminating information to all their colleagues. For this, they can 
schedule meetings with others at different time slots depending on their availability. Meetings with multiple individuals at the same time slot enables the spread of information to all participants. Naturally, the goal is to minimize the number of meetings needed to inform everyone.

Our goal is for this paper to be the inaugural effort in combining dynamic networks with a game-theoretic analysis in order to better capture the formation of real-world networks.

\subsection{Our Approach}

We initiate the study of dynamic networks from a game-theoretic perspective by  combining one of the earliest and very influential strategic network formation models for static networks, the non-cooperative network formation model by~\citet{bala2000noncooperative}, with the seminal temporal graph model of~\citet{KKK02}.   

In the network formation model by~\citet{bala2000noncooperative}, the agents are nodes of a network and they strategically create costly incident links to maximize the number of nodes they can reach either directly or via a sequence of hops in the network. The temporal graph model of~\citet{KKK02} assumes that an edge-labeled graph is given, where the labels indicate the time step where the respective edge is available. By combining the features of these two models, we assume that a temporal graph serves as the host graph for our network formation game. Agents correspond to its nodes and can create costly incident edges having the time labels specified by the host graph. Most importantly, instead of using standard reachability defined as the existence of a path between two nodes, we employ the concept of temporal reachability, where some node $u$ can reach a node $v$ if a temporal path, i.e., a path with monotonically increasing edge labels, exists.

As a first step in this line of research, we consider a restricted version, where the underlying temporal host graph is a clique, all edges have unit cost, and the objective of each agent is to create as few edges as possible to ensure the existence of a temporal path from itself to every other node of the graph. Although being the simplest variant of our framework, this setting has the striking feature that equilibrium states of our game correspond to temporal spanners of the underlying temporal host graph. Thus, our model captures the decentralized creation of a temporal spanner by selfish agents. To the best of our knowledge, so far only centralized approaches exist for this prominent algorithmic problem.

We emphasize that our framework can be generalized to much more complex settings. In particular, and similarly to (recent variants of) the well-studied Network Creation Game by~\citet{fabrikant03}, more than temporal reachability could be studied in future work. For example, the existence of short temporal paths, additional robustness guarantees, and more complicated edge cost functions.

\subsection{Our Contribution}
We explore the formation of temporal spanners by strategic agents via studying the Temporal Reachability Network Creation Game, whose equilibrium networks must be temporal spanners. 
Besides this being the first decentralized approach for computing temporal spanners, the entailed equilibria must be stable with respect to local changes of the involved nodes.  

Although we show that computing a best response strategy is NP-hard even if the host graph has a lifetime $\lifetime=2$, we nonetheless show for this case that equilibria exist and that they can be computed efficiently. The existence of equilibria remains a challenging open problem for $\lifetime\geq 3$. However, we show that in this case, deciding if a given strategy profile is an equilibrium is NP-hard. This is remarkable, as similar questions are still open for most other game-theoretic network creation models. Also, in contrast to the classical Network Formation Game on static graphs~\cite{bala2000noncooperative}, this shows that incorporating temporal graphs yields a computationally much harder model.

As our main contribution, we provide non-trivial structural properties of equilibrium networks and we exploit them to prove bounds on the Price of Anarchy (PoA), i.e., on the quality of the obtained temporal spanners. Low bounds on the PoA imply that these equilibrium spanners are close to optimal. Regarding this, we give an upper bound of $\mathcal{O}(\sqrt{n})$ on the PoA and provide a lower bound of $\Omega(\log n)$. 

Moreover, driven by the hardness of computing a best response strategy, we also investigate Greedy Equilibria (GE), that rely on very simple strategy changes. We connect them to Nash Equilibria by showing that the PoA with regard to Greedy Equilibria is at most a $\mathcal{O}(\log n)$ factor larger than the PoA with regard to Nash Equilibria. This shows that not much is lost by focusing on GEs.

All omitted details can be found in the appendix.

\subsection{Related Work}
The formation of networks by strategic agents has been studied intensively within the last decades. One of the earliest models is also closest to our work. In the Network Formation Game by~\citet{bala2000noncooperative} selfish agents buy incident edges and their utility is a function that increases with the number of agents they can reach and it decreases with the number of edges bought. Most relevant for us is the version where undirected edges are formed. For this the authors prove that equilibria always exist and that they are either stars or empty graphs. Moreover, improving response dynamics quickly converge to such states. Also, for this model computing a best response strategy and deciding if a given state is in equilibrium can be done efficiently. 

The network formation game was extended to a setting with attacks on the formed network~\cite{Goyal16}. There, the objective is post-attack reachability. This variant is more complex, but best response strategies can still be computed efficiently~\cite{FriedrichIKLNS17}. Recently, a variant with probabilistic attack was studied~\cite{Chen19}. Also related are Topology Control Game~\cite{EidenbenzKZ06}, where the agents are points in the plane and edge costs are proportional to the Euclidean distance among the endpoints. Similar in spirit is the model by~\citet{gulyas2015navigable}, but there the agents are points in hyperbolic space using greedy routing.
Also models exist where the agents aim for creating a robust network, i.e., communication in the network should rely on more than a single path~\cite{MeiromMO15,ChauhanLMM16,Echzell0LM20}.

Besides network formation games with reachability objective, even more model variants exist, where shortest path distances play a prominent role. Starting with the Network Creation Game~\cite{fabrikant03}, that is based on the even older Connections Game~\cite{jackson1996}, researchers have focused on utility functions that depend on the distances of the respective agent to the other agents in the formed network. For this, variants exist that involve cooperation~\cite{CorboP05,AndelmanFM09}, locality~\cite{BiloGLP16,Cord-LandwehrL15}, non-uniform edge prices~\cite{ChauhanLMM17,bilo2019}, and most recently, social networks~\cite{BFLLM21,FGLM22,BullingerLM22}.

For most of these variants, computing a best response strategy is NP-hard, improving response dynamics are not guaranteed to converge, and the hardness of deciding equilibria is open. Moreover, for the original Network Creation Game~\cite{fabrikant03} equilibria always exist and the Price of Anarchy (PoA) is known to be constant for almost the full parameter range of the model. Similar results hold for the other models, with some notable exceptions, e.g., the geometric version has a high PoA~\cite{bilo2019}.

To the best of our knowledge, no game-theoretic network formation model involving temporal graphs has been studied. However, starting from the work by~\citet{KKK02}, a lot of research has been devoted to algorithmic problems on temporal graphs, in particular to temporal spanners. Relevant for us, it has been shown that temporal cliques admit sparse temporal spanners~\cite{CasteigtsPS21} and also sparse spanners with low stretch are possible~\cite{BiloDG0R22}. In contrast, on non-complete temporal graphs spanners can be very dense~\cite{AxiotisF16}. The same holds true if strict temporal paths are considered, i.e., if the labels on the edges of a path must strictly increase~\cite{KKK02}. Closely related to the reachability problem, \cite{klobasMMS22} study the problem of finding the minimum number of labels required to achieve temporal connectivity in a graph.

Spanners on static graphs are a classical topic, see~\cite{AhmedBSHJKS20} for a recent survey. For spanners in geometric settings, see~\cite{narasimhan2007geometric}.

\subsection{Model and Notation}
Before stating our game-theoretic model, we will first introduce temporal graphs and temporal spanners.
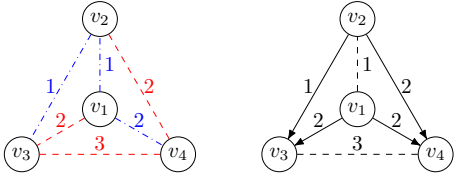
\begin{figure}
    \centering
    \scalebox{0.8}{\begin{tikzpicture}[on grid=true, node distance=1.5cm and 1.5cm]
    	\node[nod] (v1) {$v_1$};
    	\node[nod] (v2) at (90:1.5) {$v_2$};
    	\node[nod] (v3) at (210:1.5){$v_3$};
    	\node[nod] (v4) at (330:1.5){$v_4$};
    	\path[dashed, red]
    	   \ed(v1,v3,2)
    	   \ed(v3,v4,3)
    	   \edo(v2,v4,2,right);
    	\path[dash dot, blue]
    	   \edo(v2,v3,1,left)
    	   \edo(v1,v2,1, right)
    	   \ed(v1,v4,2)
            ;
     
    \end{tikzpicture}\hspace{1cm}
    \begin{tikzpicture}[on grid=true, node distance=1.5cm and 1.5cm]
    	\node[nod] (v1) {$v_1$};
    	\node[nod] (v2) at (90:1.5) {$v_2$};
    	\node[nod] (v3) at (210:1.5){$v_3$};
    	\node[nod] (v4) at (330:1.5){$v_4$};
    	\path[dashed]
    	\edo(v1,v2,1, right)
    	\ed(v3,v4,3);
     
        \path[-{Latex[round]}]
    	\ed(v1,v3,2)
    	\ed(v1,v4,2)
    	\edo(v2,v3,1,left)
    	\edo(v2,v4,2,right);
    \end{tikzpicture}}
    \caption{The left shows a temporal graph with $n=4$ nodes and lifetime $\lifetime=3$. The sequence $v_1,v_3,v_4,v_2$ (red) is not a temporal path since its labels are not monotonically increasing. On the other hand, $v_3,v_2, v_1, v_4$ (blue) is a temporal path from $v_3$ to $v_4$.\\
    The right shows the graph from the left as the host graph $H$ and the graph $G(\s)$ (not dashed) formed by the strategies of the agents. Here, $v_1$ plays greedy best response, since neither adding the edge $(v_1,v_2)$ nor removing one of the edges $(v_1,v_3)$ or $(v_1,v_4)$ decreases its cost. However, $v_1$ does not play best response since removing $(v_1,v_3)$, $(v_1,v_4)$ and buying $(v_1,v_2)$ is an improving move.}
    \label{fig:model}
\end{figure}
\paragraph*{Temporal Graphs and Spanners.}
A \emph{temporal graph} $G = (V_G,E_G,\lambda_G)$ is an undirected labeled graph, where
$\lambda_G\colon E_G\rightarrow\N$, assigns a label
to each edge. 
The edge labels of $G$ model at which point in time an edge is available. For simplicity, we assume that all labels are consecutive starting with 1. Formally, $\bigcup_{e\in E_G}\lambda_G(e)=\{1,2,\dots,\lifetime\}$, for some $\lifetime\in\N$, where $\lifetime$ is called the \emph{lifetime}. Note that, for simplicity, we assume that every edge has a single time label, i.e., every edge is available only at a particular single time step. All results in this paper also hold if we extend the model to allow multiple labels per edge.
As long as the graph $G$ is clear from context, we might omit the subscripts and write $V,E,$ and $\lambda$ instead of $V_G,E_G,$ and $\lambda_G$.

A (simple) \emph{temporal path} in $G$ is a (simple) path in $G$ with monotonically increasing edge labels. Formally, it is a sequence of (distinct) nodes $v_1, \dots,v_i \in V_G$, that forms a (simple) path in $G$, i.e., for $1\leq j \leq i-1$ we have $e_j = \{v_j,v_{j+1}\}\in E_G$, where for all $1\leq j \leq i-2$ we have $\lambda_G(e_j) \leq \lambda_G(e_{j+1})$.
Note, that the labels do not have to increase strictly since we assume zero edge traversal time.

For two nodes $u,v\in V_G$, we say that $u$ can \emph{reach} $v$ in the temporal graph~$G$ if and only if there is a temporal path from $u$ to $v$ in $G$. We define $\reach{G}(u)$ as the set of nodes that $u$ can reach in $G$. Note that $u\in\reach{G}(u)$, since every node can trivially reach itself via a temporal path of length 0. If every node can reach every other node, we say that $G$ is \emph{temporally connected}. With this we define a \emph{temporal spanner} of $G$ as any temporally connected subgraph $G'$ of $G$ with $V_{G'}=V_{G}$ and $E_{G'} \subseteq E_G$. If no edge can be removed from $G'$ while keeping the temporal spanner property, we call $G'$ a \emph{minimal temporal spanner} of $G$. If $G'$ has at most as many edges as any other temporal spanner of $G$, we call $G'$ a \emph{minimum temporal spanner} of $G$.

\paragraph*{The Temporal Reachability Network Creation Game.}
Now we define our game-theoretic network creation model, called the \emph{Temporal Reachability Network Creation Game (TRNCG)}. Let $H = (V_H,E_H,\lambda_H)$ be a given complete temporal graph that serves as the \emph{host graph} of our game.
We assume that every node $v \in V_H$ corresponds to a strategic agent and let $|V_H| = n$ denote the number of agents. 

We assume that agents play strategies, where a \emph{strategy} $S_v$ of some agent $v$ is defined as $S_v \subseteq V_H \setminus \{v\}$, i.e., the strategy specifies to which other agents agent $v$ wants to create an edge.
The strategies of all agents together form the \emph{strategy profile} $\s =\bigcup_{v\in V_H} \{(v,S_v)\}$.

A strategy profile $\s$ defines the \emph{created directed temporal graph} $G(\s) = (V_{G(\s)}, E_{G(\s)},\lambda_{G(\s)})$, with $V_{G(\s)} = V_H$, $E_{G(\s)} = \{(u,v) \mid u,v\in V_H \wedge v\in S_u\}$,
and where $\lambda_{G(\s)}$ is the labeling $\lambda_H$ restricted to the edge set $E_{G(\s)}$, with $\lambda_{G(\s)}((u,v)) = \lambda_{G(\s)}((v,u)) = \lambda_H(\{u,v\})$. Thus, $G(\s)$ is a directed subgraph of the complete host graph $H$ that contains the union of all edges that are created by the agents, i.e., for every edge in $G(\s)$ there is exactly one agent that wants to create it. We use directed edges to encode the owner of the edge, where edges are always directed away from their owner. For reachability, these edge directions will be ignored. 

Agents choose their respective strategy to minimize their individual cost, where the cost of agent $v$ in the created directed temporal graph $G(\s)$ is defined as
\begin{equation*}
    \cost{v}{\s}{H} = |S_v|+\bigconstant\cdot|V_H\setminus\reach{{\Gu(\s)}}(v)|,
\end{equation*}
where $\bigconstant>1$ is a large constant and $\Gu(\s)$ is the undirected version of $G(\s)$, i.e., $V_{\Gu(\s)} = V_{G(\s)}$ and $E_{\Gu(\s)} = \{\{u,v\}\mid u\in S_v \vee v \in S_u\}$ with $\lambda_{\Gu(\s)}(\{u,v\}) = \lambda_H(\{u,v\})$, for all edges $\{u,v\} \in E_{\Gu(\s)}$.  
Thus, in the created temporal graph $G(\s)$, agent $v$ incurs a cost of one unit for each edge it creates and a penalty of $\bigconstant$ for each agent it cannot reach via a temporal path that ignores edge directions. Hence, agents aim to create as few edges as possible while still maintaining undirected temporal reachability. 

Let $\s_{-v} = \s\setminus(v,S_v)$ denote the set of strategies of all agents other than agent $v$. Now consider that $v$ changes its strategy $S_v$ to $S_v'$. The resulting strategy profile is $\s_{-v}\cup \{(v,S_v')\}$ which we will abbreviate as $\s_{-v}\cup S_v'$. We say that agent $v$'s strategy change from $S_v$ to $S_v'$ is an \emph{improving move}, if it yields strictly less cost for agent $v$, i.e., if $\cost{v}{\s_{-v}\cup S_v'}{H} < \cost{v}{\s}{H}$. If we additionally restrict the strategy change to a single addition or deletion\footnote{Formally, for $x\in S_v$ and $y\in V_H \setminus S_v$, we have $S_v'=S_v\setminus\{x\}$ or $S_v'=S_v\cup\{y\}$. Note, that typically in the literature, greedy improving moves also allow swaps, i.e. $S_v'=(S_v\setminus \{x\})\cup \{y\}$. However, due to our model definition we can ignore swaps because if swapping $x$ for $y$ is improving, simply adding $y$ is also improving.}, we call it a \emph{greedy improving move}. For an example, see \Cref{fig:model}. If there is no improving move or greedy improving move of $v$ for $\s$, we call $S_v$ a \emph{best response} or \emph{greedy best response}, respectively.

Given this definition, we can now define our solution concepts. For a given host graph $H$ we say that the strategy profile $\s$ is in \emph{Pure Nash Equilbrium (NE)}~\cite{nash50}, if no agent has an improving move. We say $\s$ is in \emph{Greedy Equilibrium (GE)}~\cite{Lenzner12} if no agent has a greedy improving move. Note, that every greedy improving move is also an improving move and therefore every NE is also in GE.

Since we have a bijection between $\s$ and $G(\s)$, we will use the strategy profile $\s$ and its corresponding created graph $G(\s)$ interchangeably and we will say that $G(\s)$ is in NE (or GE).
Note that for any graph $G(\s)$ in NE or GE, it follows that for every edge $(u,v) \in E_{G(\s)}$ we have $(v,u)\notin E_{G(\s)}$. This is true, since otherwise one of the agents could omit the other from its strategy without removing the edge from $\Gu(\s)$ and thereby decrease its cost.

Given a temporal host graph $H$ and a strategy profile $\s$, we want to compare different created graphs in terms of their social cost. Here, the \emph{social cost} of some created graph $G(\s)$ is defined as
\begin{equation*}
    \SC{\s}{H} = \sum_{\mathclap{v\in V_H}}\cost{v}{\s}{H}=|E_{G(\s)}|+\bigconstant\sum_{\mathclap{v\in V_H}}|V_H\setminus\reach{\Gu(\s)}(v)|.
\end{equation*}
 Note that the social cost of $G(\s)$ equals $|E_{G(\s)}|$ if $\Gu(\s)$ is a temporal spanner. If for some host graph $H$ the strategy profile $\s_H^*$ minimizes the social cost, we call $\s_H^*$ a \emph{social optimum for $H$} and, by extension, the corresponding graph $G(\s_H^*)$ a \emph{social optimum subgraph of $H$}, often denoted as \opt.
 Since $K$ is large, the set of social optimum subgraphs and the set of minimum temporal spanners for $H$ coincide.
 
To investigate the efficiency loss from letting agents act selfishly towards minimizing their costs, we define the \emph{Price of Anarchy (PoA)}~\cite{KP99}.
The PoA is the worst ratio between the social cost of any stable state and the social cost of the corresponding social optimum on the same host graph.
Towards a formal definition, let $\textnormal{NE}_H$ denote the set of strategy profiles that are in NE for a given host graph~$H$. Moreover, let $\mathcal{H}_{n,\maxLifetime}$ denote the set of all possible complete temporal host graphs with $n$ nodes and a lifetime of at most $\maxLifetime$. Then, the PoA, with respect to NE, is formally defined as
$\poa_{\textnormal{NE}}(n,\maxLifetime) = \sup_{H\in\mathcal{H}_{n,\maxLifetime}}\max_{\mathbf{s}\in \textnormal{NE}_H}\frac{\SC{\mathbf{s}}{H}}{\SC{\mathbf{s}_H^*}{H}}$. The PoA w.r.t. GE is defined analogously. We write $\poa_{\textnormal{NE}}(n)$ instead of $\poa_{\textnormal{NE}}\left(n,{n \choose 2}\right)$ when we consider host graphs with arbitrary lifetime.

Lastly, we define an important property of the game dynamics: We say that the TRNCG has the \emph{finite improvement property} if any sequence of improving moves must be finite.

\section{Computational Complexity}
In this section, we prove that computing best responses and checking strategy profiles whether they are in NE is NP-hard. This is a significant difference to network formation games with reachability objective where computing best responses is computationally easy~\cite{bala2000noncooperative}. This means that adding the temporal component to the model makes it considerably harder.
\begin{restatable}{theorem}{bestResponseNPHard}
\label{thm:improving_move_np_hard}
    Given a tuple $(H,\s,x)$ consisting of a complete temporal host graph $H$ with lifetime $\lifetime\ge 2$, a strategy profile $\s$, and a node $x\in V_H$, computing a best response for $x$ is NP-hard.
\end{restatable}
\begin{figure}
    \centering
    \includegraphics[scale=0.73]{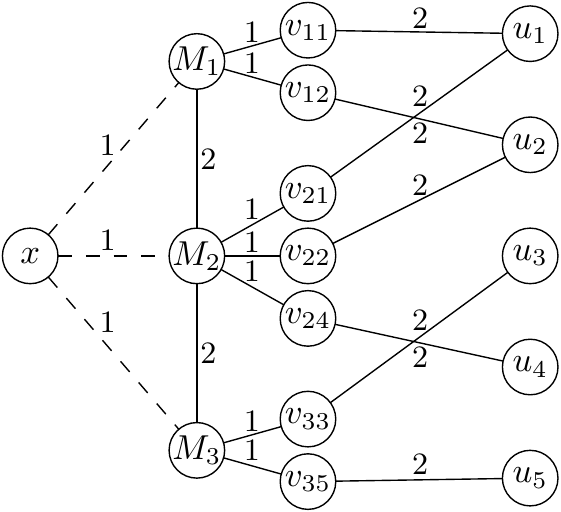}
    \includegraphics[scale=0.73]{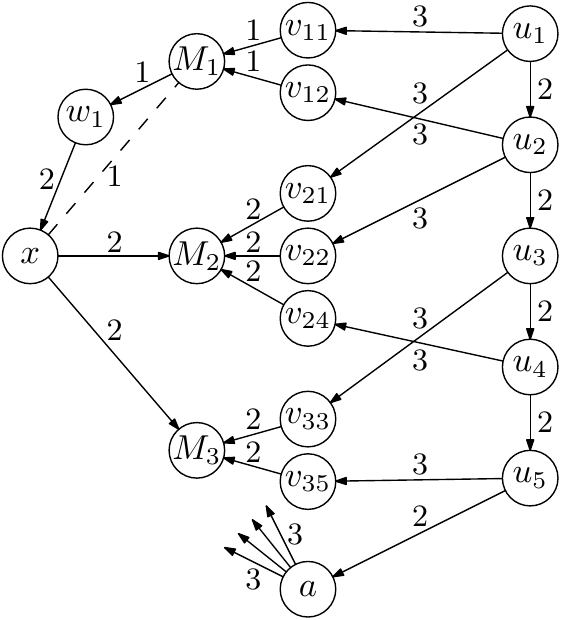}
    \caption{This figure shows examples of the constructions for~\Cref{thm:improving_move_np_hard} (left) and~\Cref{thm:NE_np_hard} (right), given the set cover instance consisting of the universe $U\coloneqq\{u_1,\dots,u_5\}$ and the sets $\mathcal{M}\coloneqq\{M_1,\dots,M_3\}$ with $M_1\coloneqq\{u_1,u_2\}$, $M_2\coloneqq\{u_1,u_2,u_4\}$ and $M_3\coloneqq\{u_3,u_5\}$. Additionally, on the right, we are given a set cover consisting of $M_2$ and $M_3$ but not $M_1$. This is encoded by $x$ buying edges towards $M_2$ and $M_3$ and the existence of $w_1$.}
    \label{fig:np_hard}
\end{figure}
\begin{theorem}\label{thm:NE_np_hard}
    Deciding whether a pair $(H, \s)$ consisting of a complete temporal host graph $H$ with lifetime $\lifetime\geq 3$ and a strategy profile $\s$ is a NE is NP-hard.
\end{theorem}
\begin{proof}
    Given a universe $U\coloneqq\{u_1,\dots,u_k\}$ of $k$ elements and a set of $m$ sets $\mathcal{M}\coloneqq\{M_1,\dots,M_m\}\subseteq\Pow(U)$, deciding whether a given set cover $\mathcal{C} \subseteq \mathcal{M}$ for $U$ is a minimum set cover is NP-hard. We give a polynomial time reduction that, given an instance $(U,\mathcal{M},\mathcal{C})$, constructs a tuple $(H,\s)$ consisting of a temporal host graph $H$ and a strategy profile $\s$. We show that $\mathcal{C}$ is a minimum set cover if and only if $\s$ is a NE for $H$. Instead of defining $\s$ directly, we define $G\coloneqq G(\s)$. More precisely, the host graph $H$ is defined as follows
    \begin{align*}
        V_H&\coloneqq\{x,a\}\cup\mathcal{M}\cup U\cup\bigcup_{i=1}^m\bigcup_{u_j\in M_i}\{v_{ij}\}\cup \bigcup_{M_i\in \mathcal{M}\setminus \mathcal{C}}\{w_{i}\}\\
        \lambda(e)&\coloneqq
        \begin{cases}
            1&\text{if } \exists M_i \in \mathcal{M}\setminus\mathcal{C},j\in \N\colon\\
            &M_i \in e \land \{x,w_i,v_{ij}\} \cap e \neq \emptyset\\
            2&\text{if } \exists M_i \in \mathcal{C},j\in \N\colon\\
            &(M_i \in e \land \{x,v_{ij}\} \cap e \neq \emptyset)\\
            &\lor e\in\{\{w_j,x\},\{u_n,a\},\{u_j,u_{j+1}\}\}\\
            3&\text{otherwise.}
        \end{cases}
    \end{align*}
    The set of edges of $G(\s)$ is defined as follows
    \begin{align*}
    E_G&\coloneqq\bigcup_{i=1}^{n-1}\{(u_i,u_{i+1})\}\cup\bigcup_{i=1}^m\bigcup_{u_j\in M_i}\big\{(v_{ij},M_i),(u_j,v_{ij})\big\}\\*
        &\quad\cup\big\{(u_n,a),(a,x)\big\}\cup \bigcup_{i=1}^{m}{\{(a,M_i)\}}\cup \bigcup_{M_i\notin C}{\{(a,w_i)\}}\\*
        &\quad\cup \bigcup_{M_i\in C}{\{(x,M_i)\}} \cup\bigcup_{M_i\notin C}{\big\{(M_i,w_i),(w_i,x)\big\}}.
    \end{align*}
    Intuitively, we construct a node for each set in $\mathcal{M}$ and each element in $U$ and connect each set with all its elements via a monotonically increasing path of length 2. See~\Cref{fig:np_hard} for an example of the construction. The other edges are chosen so that $G$ is a temporal spanner and all nodes except for $x$ play best response. This can easily be checked for every node. Hence, to check whether $s$ is a NE, we only need to check whether $x$ plays best response.
    
    Let $\mathcal{C}'\subseteq \mathcal{M}$ be a minimum set cover for $U$ and $S_x^b$ a best response of $x$ for $\s$.
    
    Consider $S_x'\coloneqq\mathcal{C}'$, meaning that $x$ builds all the edges $\{x,M_i\}$ for $M_i\in\mathcal{C}'$. We see that $x$ can now reach every node in $G(\s_{-x}\cup S_x')$. Therefore, when $\mathcal{C}$ is not a minimum set cover and therefore $|\mathcal{C}'| < |\mathcal{C}|$ it follows that $|S_x'| < |S_x|$ which implies that $\s$ is not a NE.
    
    Since $S_x^b$ is a best response, $x$ can reach every node in $G(\s_{-x}\cup S_x^b)$. Suppose, $x$ buys an edge to one of the nodes $u_j$ or $v_{ij}$. Instead, $x$ can buy an edge to a node $M_i$, such that $u_j\in M_i$, without breaking reachability. Therefore, there is a best response ${S_x^b}'\subseteq\mathcal{M}$. Note that $x$ still reaches all nodes $u\in U$ in $G(\s_{-x}\cup {S_x^b}')$ and $u$ can only be reached by $x$ if there is $M\in\mathcal{M}$ such that $x$ builds an edge to $M$. This means that $\mathcal{C}''\coloneqq {S_x^b}'$ is a set cover. Therefore, if $\s$ is not a NE, it follows $|S_x^b|<|S_x|$ which implies that $|\mathcal{C}''| < |\mathcal{C}|$, so $\mathcal{C}$ is not a minimum set cover.
    
    It is obvious that this construction is computable in polynomial time which concludes the proof.
\end{proof}

While NE is a very natural solution concept, the fact that it is computationally hard to compute best responses raises the question of whether it can realistically model the selfishness of the agents. Therefore, we will also consider GE because greedy best responses are computable in polynomial time.
\begin{proposition}
    Given a tuple $(H,\s,x)$ consisting of a temporal host graph $H$, a strategy profile $\s$, and a node $x\in V_H$,  a greedy best response for $x$ is computable in polynomial time.
\end{proposition}

\section{Existence and Properties of Equilibria}
We discuss under what circumstances equilibria exist and what properties they have. We start by showing that equilibrium existence cannot be proven via potential functions.
\begin{restatable}{theorem}{potentialGame}
    The TRNCG is not a potential game.
\end{restatable}

Note that all the strategy changes in the improving response cycle are greedy best responses, too. This means that it is not a potential game even with regard to GE.

Next, we show that equilibria always exist when the lifetime is $\lifetime=2$ and that we can find one in polynomial time. This contrasts the result from~\Cref{thm:improving_move_np_hard} which showed that, even for $\lifetime=2$, computing best responses is NP-hard.

\begin{theorem}
    Let $H$ be a complete temporal host graph with $\lifetime=2$. Then there is a strategy profile $\s$ in NE for $H$.
\end{theorem}
\begin{proof}
    We show this by proving that $H$ contains a spanning tree $T$ whose edges all have the same label. Note that any strategy profile $\s$ such that $\Gu(\s)=T$ is a NE.

    Let $H_i$ denote the subgraph of $H$ on $V_H$ that contains all edges of $H$ of label $i$. We show that at least one between $H_1$ and $H_2$ is connected, thus proving the existence of $T$. 
    
    The claim trivially follows if $H_1$ is connected. So, assume that $H_1$ is not connected, i.e., there is a cut $(X,Y)$ that is traversed by none of the edges in $H_1$. All the edges that traverse the cut $(X,Y)$ are in $H_2$.
    Hence, $H_2$ is connected.
\end{proof}

In the following, we prove that stable graphs cannot contain too many edges. We first bound the number of edges for graphs with a small lifetime $\lifetime$.
\begin{theorem}\label{lem:lifetime_stable}
    Let $H$ be a complete temporal host graph containing $n>2$ nodes and lifetime $\lifetime>1$. Then any GE contains at most $\lifetime(n-2)$ edges.
\end{theorem}
\begin{proof}
    Let $\s$ be a strategy profile and $G\coloneqq G(\s)$ a GE. If there are at least $n$ edges with some label $l$, some of them form a cycle. Removing one edge from the cycle does not change reachability among pairs of nodes. Therefore, each label can appear at most $n-1$ times. Furthermore, if there are $n-1$ edges and no cycles with the same label in $G$, those edges would form a spanning tree of $G$ that guarantees reachability among pairs of nodes. Then, no other edge would be needed. Therefore, if $G$ contains at least two labels, each label appears at most $(n-2)$ times. Combined with only $\lifetime$ labels existing, $G$ contains at most $\lifetime(n-2)$ edges.
\end{proof}

Next, we prove an upper bound on the number of edges in an equilibrium state independent of $\lifetime$. We start by introducing the concept of \emph{necessary edges} which we then use to characterize a structure that cannot appear in an equilibrium. 
\begin{definition}[necessary edge]
    Let $H$ be a complete temporal host graph with $n$ agents, $\s$ a strategy profile and $G\coloneqq G(\s)$. For each edge $e=(u,v)\in G$ that $u$ buys, we define
    $$A_G(e)\coloneqq \big\{x\in V\mid x\in \reach{G}(u)\wedge x\notin\reach{G-e}(u)\big\}.$$
    We say that $e$ is \emph{necessary} for $u$ to reach the agents in $A_G(e)$.
\end{definition}
Note, that if $G$ is a GE, $A_G(e)\neq\varnothing$ for all $e\in E_G$.

Using this definition, we characterize a structure that cannot appear in any strategy profile.
\begin{lemma}\label{lem:forbidden_structure}
    Let $H$ be a complete temporal host graph, $\s$ be a strategy profile, and $G\coloneqq G(\s)$. There cannot be nodes $z,u_1,u_2,x,y\in V$ and distinct edges $e_{1x}, e_{1y}, e_{2x}, e_{2y}\in E_G$ such that for $i\in\{1,2\}$ and $j\in\{x,y\}$
    \begin{compactenum}
        \item $\{z,u_i\}\in E_{\Gu}\setminus\{e_{ij}\}$;
        \item $e_{ij}$ is bought by $u_i$;
        \item $j\in A_G(e_{ij})$; 
        \item $\lambda(\{z,u_i\})\leq\lambda(e_{ij})$. \hfill (See~\Cref{fig:forbidden_structure}.)
    \end{compactenum}
    
\end{lemma}
\begin{figure}[h]
    \centering
    \includegraphics[scale=0.9]{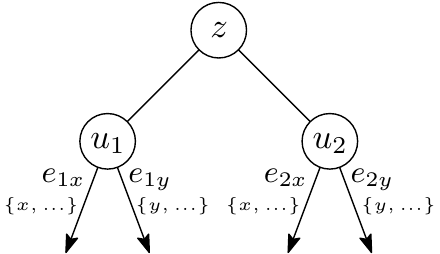}
    \caption{A forbidden structure in a strategy profile. The node $z$ has two neighbors $u_1$ and $u_2$ that both have two distinct edges that they need to reach the nodes $x$ and $y$ respectively. For both of them, the two needed edges have at least the same label as their edge to $z$.}
    \label{fig:forbidden_structure}
\end{figure}
\begin{proof}
    Towards a contradiction, suppose there are nodes $z,u_1,u_2,x,y\in V$ and edges $e_{1x}, e_{1y}, e_{2x}, e_{2y}\in E_G$ as defined above. W.l.o.g. $\lambda(\{z,u_1\})\le\lambda(\{z,u_2\})$ and $\lambda(e_{1x})\le\lambda(e_{1y})$. Any temporal path from $u_2$ to $x$ starts with $e_{2x}$ (since $x\in A_G(e_{2x})$) and has to use $e_{1x}$. Otherwise, $u_1$ could reach $x$ by using $\{z,u_1\}$, $\{z,u_2\}$ and then the path from $u_2$ to $x$ without needing $e_{1x}$ which contradicts $x\in A_G(e_{1x})$. The same holds for $y$ instead of $x$.
    
    Therefore, there is a temporal path $P$ from $u_2$ to $u_1$, which starts with $e_{2x}$ and arrives at $u_1$ no later than $\lambda(e_{1x})\le\lambda(e_{1y})$. This means that $u_2$ does not need $e_{2y}$ to reach $y$ since it can use $P$ to get to $u_1$ and travel to $y$ from there. This contradicts $y\in A_G(e_{2y})$.
\end{proof}
The forbidden structure from~\Cref{fig:forbidden_structure} implies that graphs with at least $\sqrt{6}n^\frac{3}{2}+n$ edges must contain unnecessary edges, giving us a bound on the number of edges in an equilibrium.
\begin{theorem}\label{thm:dense_not_ge}
    Let $H$ be a complete temporal host graph with $|V_H|=n$ agents and $\s$ be a strategy profile. If $G\coloneqq G(\s)$ contains at least $\sqrt{6}n^\frac{3}{2}+n$ edges, then $G$ is not a GE.
\end{theorem}
\begin{proof}
    Towards a contradiction, suppose that $G$ is a GE and $|E_G|\ge\sqrt{6}n^\frac{3}{2}+n$. As shown in~\iflong\Cref{lem:large_node_exists} \else Lemma 15 \fi in the appendix, there is a node $z \in V_G$ and a set $M\subset V_G$ such that
    \begin{compactenum}
        \item $|M| = \lceil\frac{1}{3}\sqrt{6n}\rceil$;
        \item $(u,z) \in E_G$, for every $u \in M$;
        \item Each $u\in M$ has a set $E_u\subseteq E_G$ of at least $\frac{2}{3}\sqrt{6n}$ outgoing edges $(u,v)$ with $z\neq v$ and $\lambda_G((u,z))\le\lambda((u,v))$.
    \end{compactenum}
    See \Cref{fig:graph_including_forbidden_structure} for an illustration.
    \begin{figure}[h]
    \centering
    \includegraphics[scale=0.9]{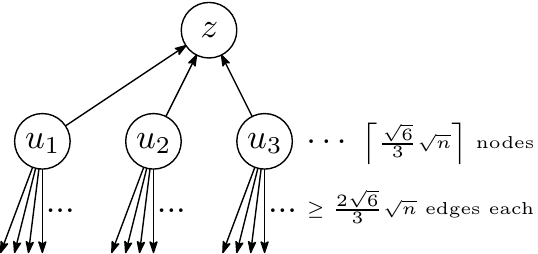}
    \caption{A structure that always appears in a directed temporal graph with at least $\sqrt{6}n^\frac{3}{2}+n$ edges. A node $z$ exists with $\lceil\frac{\sqrt{6}}{3}\sqrt{n}\rceil$ neighbors via in-edges that each have at least $\frac{2\sqrt{6}}{3}\sqrt{n}$ out-edges with a label that is at least as high as the label of their edge to~$z$.}
    \label{fig:graph_including_forbidden_structure}
\end{figure}

    For each edge $e\in E_u$, let $a_e\in A_G(e)$ be a representative of $A_G(e)$. Note that $A_G(e)\neq\varnothing$ because $G$ is a GE, so those representatives always exist.
    For each $u\in M$, we define
    \begin{equation*}
        D_u\coloneqq\bigcup_{e\in E_u}\{a_e\}.
    \end{equation*}
    Intuitively, $D_u$ contains nodes that $z$ can reach by going over $u$ and that $u$ needs to buy an edge for.
    
    We see that the forbidden structure from~\Cref{lem:forbidden_structure} appears if there are two nodes $u,v\in M$ such that $|D_u\cap D_v|\ge 2$. We can therefore assume $|D_u\cap D_v|\leq 1$ for all $u,v \in M$. Also, for $e,e'\in E_u$ we have $A_G(e)\cap A_G(e')=\varnothing$ since there cannot be two edges that are necessary for $u$ to reach the same node. From this, we get $|D_u|\ge |E_u|\ge\frac{2}{3}\sqrt{6n}$.
    
    Using the inclusion-exclusion principle, we get
    \resizebox{\linewidth}{!}{\begin{minipage}{\linewidth}
    \begin{align*}
        \left|\bigcup_{u\in M}D_u\right|&\ge\sum_{u\in M}|D_u|-\smashoperator[r]{\sum_{\{u,v\}\subseteq M,u\neq v}}|D_u\cap D_v|\\
        &\ge \left\lceil\frac{1}{3}\sqrt{6n}\right\rceil\frac{2}{3}\sqrt{6n}-\frac{1}{2}\left\lceil\frac{1}{3}\sqrt{6n}\right\rceil\left(\left\lceil\frac{1}{3}\sqrt{6n}\right\rceil-1\right)\\
        &> \left\lceil\frac{1}{3}\sqrt{6n}\right\rceil\frac{2}{3}\sqrt{6n}-\frac{1}{2}\left\lceil\frac{1}{3}\sqrt{6n}\right\rceil\frac{1}{3}\sqrt{6n}\\
        &= \left\lceil\frac{1}{3}\sqrt{6n}\right\rceil \frac{1}{2}\sqrt{6n} \geq n.
    \end{align*}
    \end{minipage}}\\
    
    This is a contradiction since $G$ has only $n$ nodes.
\end{proof}

\section{Quality of Equilibria}
Finally, we characterize the quality of equilibra by analyzing the Price of Anarchy (PoA). We show that the PoA with regard to NE and GE are within a $\log(n)$-factor of each other, meaning that any bound on the easier to analyze Greedy Equilibria also yields a bound for Nash Equilibria. We the proceed to give several upper and lower bounds.

For analyzing the PoA, we start by upper bounding the cost of the social optimum. The following result follows from~~\cite{CasteigtsPS21}.
\begin{restatable}{theorem}{theoremUBsocopt}\label{thm:upper_bound_social_optimum}
    Let $H$ be a complete temporal host graph with $|V_H|=n$ agents and $\opt$ be a social optimum for $H$. Then $\SC{\opt}{H}=|E_{\opt}|\in\mathcal{O}(n\log(n))$.
\end{restatable}
Whether we consider NE or GE has an impact on the PoA. However, we show that these two PoA's differ only by at most a factor of $O(\log(n))$.
\begin{theorem}
    $\poa_{\textnormal{GE}}(n)\le \mathcal{O}(\log(n))\poa_{\textnormal{NE}}(n)$.
\end{theorem}
\begin{proof}
    Let $H$ be a complete temporal host graph with $n$ vertices and lifetime $\lifetime$. Moreover, let $\s$ be a strategy profile in GE and let $\s_H^*$ the social optimum for $H$ such that $\poa_{\textnormal{GE}}(n)=\frac{\SC{\s}{H}}{\SC{\s_H^*}{H}}$. We construct a new temporal host graph $H'$ by relabeling all edges not in $G(\s)$ to $\lifetime+1$.
    
    First, we argue that $\s$ is in NE with respect to $H'$. Consider a strategy change for a node $u$ that removes $k$ edges and adds $l$ edges. Since $\mathbf{s}$ is stable against edge removal, removing those $k$ edges will lead to at least $k$ nodes no longer being reachable from $u$. Additionally, agent $u$ can use the $l$ added edges only to reach the other endpoints of the edges since all edges not in $G(\mathbf{s})$ have a higher label in $H'$ than all the edges in $G(\mathbf{s})$. Therefore, for agent $u$ to reach everyone after the strategy change, we need $k\le l$ which means that the strategy change is not an improving move.
    
    We also have $\SC{\s}{H}=\SC{\s}{H'}$ and $\SC{\s_H^*}{H}\ge n-1$. Let $\s_{H'}^*$ be a social optimum for $H'$. By~\Cref{thm:upper_bound_social_optimum} we have $\SC{\s_{H'}^*}{H'} \in\mathcal{O}(n\log (n))$, from which we get $\SC{\s_H^*}{H}\in \Omega\left(\frac{\SC{\s_{H'}^*}{H'}}{\log(n)}\right)$ and therefore
    \begin{align*}
        \frac{\SC{\s}{H}}{\SC{\s_H^*}{H}}\le
        \frac{\mathcal{O}(\log (n))\SC{\s}{H'}}{\SC{\s_{H'}^*}{H'}}\le \mathcal{O}(\log (n))\poa_{\textnormal{NE}}(n)
    \end{align*}
    which proves the claim.
\end{proof}
When considering the PoA with respect to a fixed maximum lifetime $\maxLifetime$ instead of a fixed $n$, we get an upper bound for the PoA of $\maxLifetime$.
\begin{restatable}{theorem}{PoALifetime}\label{thm:life_time_poa}
    For fixed maximum lifetime $\maxLifetime$, we have $\poa_{\textnormal{GE}}(n,\maxLifetime)\leq \maxLifetime-\frac{\maxLifetime}{n-1}$. For $\maxLifetime=2$, we have $\poa_{\textnormal{GE}}(n,2)=\poa_{\textnormal{NE}}(n,2)=2-\frac{2}{n-1}$.
\end{restatable}
\begin{proof}[Proof sketch]
    The upper bound on $\poa_{\textnormal{GE}}(n,\maxLifetime)$ follows directly from~\Cref{lem:lifetime_stable} by lower bounding the social cost of the optimum with $n-1$.
    
    For the tight lower bound when $\maxLifetime=2$ which also holds for $\poa_{\textnormal{NE}}(n,2)$, we construct the graph from \Cref{fig:lower_bound_constructions}.
\end{proof}
\begin{figure}[h]
    \centering
    \raisebox{0.35\height}{\includegraphics[scale=0.8]{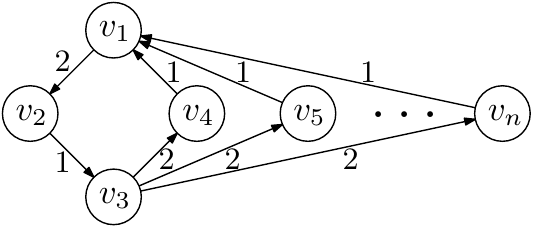}}
    \hspace{0.3cm}
    \includegraphics[scale=0.7]{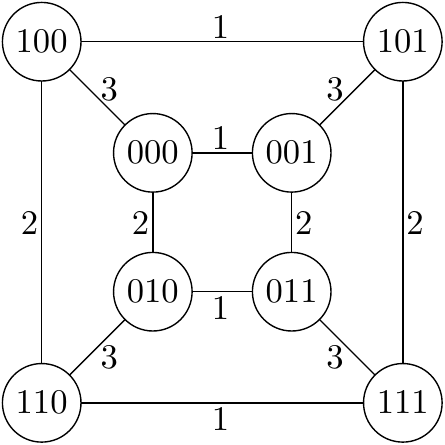}
    \caption{Left: A NE containing $2(n-2)$ edges and only 2 labels. All edges not taken have label 2. Right: A 3-dimensional hypercube. Each node corresponds to a bitstring of length 3. An edge exists between two bitstrings if they differ in exactly one position. The label of an edge is the position that the incident bistrings differ in.}
    \label{fig:lower_bound_constructions}
\end{figure}
We show that the PoA scales at least logarithmically with $n$ by giving a class of host graphs that have equilibria which are by a logarithmic factor worse than the social optimum.
\begin{theorem}
    It holds that $\poa_{\textnormal{NE}}(n)\in \Omega(\log n)$.
\end{theorem}
\begin{proof}
    Let $n\ge 8$ be a power of 2.
    \citet{KKK02} proved how to label a $\log(n)$-dimensional hypercube with labels in $\{1,\ldots, \log(n)\}$ so as the resulting temporal graph is a minimal temporal spanner. The idea of our proof is to define the host graph $H$ on $n$ nodes so that it contains the $\log(n)$-dimensional hypercube which is a NE regardless of the strategy profile that generates it, and a spanning tree formed by edges having the same label $1+\log(n)$.
    
     Formally, the nodes are bitstrings of length $\log(n)$. The complete host graph $H$ is defined as follows
     \resizebox{\linewidth}{!}{\begin{minipage}{\linewidth}
     \begin{align*}
        V_H&\coloneqq \big\{b_1b_2...b_{\log(n)}\mid \forall 1\le i \le\log(n)\colon b_i\in\{0,1\}\big\}\\
        \lambda(\{u,v\})&\coloneqq
        \begin{cases}
            i&\text{if } u \text{ and } v \text{ differ only in position } i;\\
            \log(n)+1&\text{otherwise}.
        \end{cases}
    \end{align*}
    \end{minipage}}\\
    The strategy profile $\s$ induces a graph $G\coloneqq \Gu(\s)$ with
    \[
        E_{\Gu} \coloneqq\big\{\{u,v\}\mid v \text{ and } u \text{ differ in exactly one bit}\big\}.
    \]
    \Cref{fig:lower_bound_constructions} shows an illustration of $G$ for $n=8$.
    
    Since $n\ge 8$, $H$ contains a spanning tree $\opt$ whose edges are all labelled with $\log(n)+1$. Since $\opt$ is a temporal spanner, it is a social optimum with $\SC{\opt}{H}=n-1$.
    
    $G$ is a hypercube graph containing $\frac{n}{2}\log(n)$ edges and therefore $\SC{G}{H}=\frac{n}{2}\log(n)$. If $G$ is a NE, we get the desired bound
    \begin{equation*}
        \poa_{\textnormal{NE}}(n)\ge\frac{\SC{G}{H}}{\SC{\opt}{H}}=\frac{\frac{n}{2}\log(n)}{n-1}\in\Omega(\log n).
    \end{equation*}
    
    Finally, we argue that $G$ is a NE. First, we observe that
    there is exactly one temporal path from every node $u$ to every other node $v$. When considering the bit strings of these two nodes, the path flips all of the bits that are different in $u$ and $v$ in ascending order. By definition of $H$, all of those edges exist and are labeled in ascending order.

    Secondly, note that all edges $\{u,v\}\in E_{\Gu}$ are needed by both of their end points to reach each other. By definition of $E_G$, the bitstrings of $u$ and $v$ differ only in one position $p$. When removing the edge $\{u,v\}$, a temporal path starting from $u$ has to flip another bit than $p$. In a temporal path on $G$, this bit can never be flipped back, so $v$ cannot be reached. We also observe that buying other edges outside of $E_G$ cannot replace any edge in $E_G$.

    Hence, $G$ is a temporal spanner and no agent can improve their strategy, which makes $\s$ a $\textnormal{NE}$.
\end{proof}

An upper bound for the $\poa$ follows from \Cref{thm:dense_not_ge}.

\begin{corollary}
    $\poa_{\textnormal{GE}}(n)\in\mathcal{O}(\sqrt{n})$.
\end{corollary}
\begin{proof}
    This follows from~\Cref{thm:dense_not_ge} and by lower bounding the social cost of the optimum with $n-1$.
\end{proof}

\section{Conclusion and Outlook}
In this paper, we combine game-theoretic network creation with temporal graphs. To this end, we defined and analyzed the Temporal Reachability Network Creation Game.

Even though we consider a restricted setting with unit cost on each edge and a complete host graph, we show NP-hardness for computing best responses and for deciding NE, showing that adding temporal aspects to the model makes it much harder. As our main contribution, we show non-trivial structural properties of equilibria and use them to derive several upper and lower bounds on the Price of Anarchy.

Since the upper bound of $\mathcal{O}(\sqrt{n})$ on the PoA only uses one local property, we believe that the PoA is closer to our lower bound of $\Omega(\log(n))$. Another important open question is settling the existence of equilibria for all complete temporal host graphs. We conjecture that equilibria exist, but, based on our efforts, even proving this for lifetime $\lifetime = 3$ is challenging.

We laid the groundwork for future research in this field. There are many natural extensions of our model. As far as agent strategy is concerned, the agents might want to minimize the distance to all others or, due to the time attribute introduced by our model, the agents may want to minimize their arrival time at the other agents. Also  structural properties of the host graph could be altered. For enhanced realism, the edges could be directed, have non-uniform buying costs, and/or non-instant traversal times. The rules of the game can also be adjusted, for example by allowing cooperation.

\newpage

\appendix

\bibliographystyle{named}
\bibliography{ijcai23}

\iflong

\section{Omitted Proofs}

\bestResponseNPHard*
\begin{proof}
    We use the fact that given a universe $U\coloneqq\{u_1,\dots,u_k\}$ of $k$ elements and a set of $m$ sets $\mathcal{M}\coloneqq\{M_1,\dots,M_m\}\subseteq\Pow(U)$, it is NP-hard to compute a \emph{minimum set cover} for $U$, i.e., a minimum-size subset $\mathcal{C} \subseteq \mathcal{M}$ such that $\bigcup_{M \in \mathcal{C}}M = U$. We give a polynomial-time reduction that, given an instance ($U,\mathcal{M})$, constructs a tuple $(H,\s,x)$ consisting of a temporal host graph $H$, a strategy profile $\s$ and a node $x\in V_H$. We show that the number of edges $x$ needs to buy in a best response is then the same as the size of a minimum set cover. Instead of defining $\s$ directly, we define $\Gu\coloneqq \Gu(\s)$. Note that we use the undirected version here because for the best response of $x$, it does not matter who buys the edges in the graph that are not incident to $x$. The host graph $H$ is defined as follows
    \begin{align*}
        V_H&\coloneqq\{x\}\cup\mathcal{M}\cup U\cup\bigcup_{i=1}^m\bigcup_{u_j\in M_i}\{v_{ij}\}\\
        \lambda(e) &\coloneqq
         \begin{cases}
            1&\text{if } x\in e\vee\exists i\exists u_j\in M_i\colon e=\{M_i,v_{ij}\};\\
            2&\text{otherwise}.
        \end{cases}
    \end{align*}
     The graph $\Gu(\s)$ has the following set of edges
     \[
         E_{\Gu}\coloneqq\bigcup_{i=1}^{m-1}\{M_i,M_{i+1}\}\cup\bigcup_{i=1}^m\bigcup_{u_j\in M_i}\big\{\{M_i,v_{ij}\},\{v_{ij},u_j\}\big\}.\\
    \]
    
    Intuitively, we construct a node for each set in $\mathcal{M}$ and each element in $U$ and connect each set with all its elements via a monotonically increasing path of length 2. See~\Cref{fig:np_hard} for an example of the construction.
    
    Let $\mathcal{C}\subseteq \mathcal{M}$ be a minimum set cover for $U$ and $S_x^b$ a best response of $x$ for $\s$.
    Consider $S_x'\coloneqq \mathcal{C}$, meaning that $x$ builds all the edges $(x,C)$ for $C\in\mathcal{C}$. Agent $x$ can now reach every node in $G(\s_{-x}\cup S_x')$. Therefore, we have $|\mathcal{C}|=|S_x'|\ge |S_x^b|$.
    
    Since $S_x^b$ is a best response, $x$ can reach every node in $G(\s_{-x}\cup S_x^b)$. Suppose, $x$ buys an edge to one of the nodes $u_j$ or $v_{ij}$. Instead, $x$ can buy an edge to a node $M_i$, such that $u_j\in M_i$, without breaking reachability. Therefore, there is a best response ${S_x^b}'\subseteq\mathcal{M}$. Note that $x$ still reaches all nodes $u\in U$ in $G(\s_{-x}\cup {S_x^b}')$ and $u$ can only be reached by $x$ if there is $M\in\mathcal{M}$ such that $x$ builds an edge to $M$. This means that $\mathcal{C}'\coloneqq {S_x^b}'$ is a set cover and therefore $|S_x^b|=|{S_x^b}'|=|\mathcal{C}'|\ge|\mathcal{C}|$.
    
    It is obvious that this construction is computable in polynomial time which concludes the proof.
\end{proof}

\potentialGame*
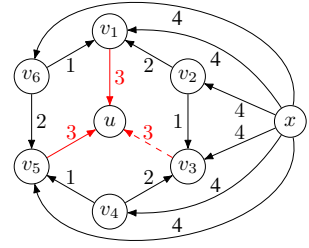
\begin{figure}
    \centering
    \scalebox{0.8}{
    \begin{tikzpicture}[on grid=true, node distance=1.5cm and 1.5cm]
        \useasboundingbox (-2,-2) rectangle (3.5,2);
        \node[nod] (u) {$u$};
    	\node[nod] (v1) at (90:1.5){$v_1$};
    	\node[nod] (v2) at (30:1.5) {$v_2$};
    	\node[nod] (v3) at (330:1.5){$v_3$};
    	\node[nod] (v4) at (270:1.5){$v_4$};
    	\node[nod] (v5) at (210:1.5){$v_5$};
    	\node[nod] (v6) at (150:1.5){$v_6$};
        \node[nod] (x) at (3,0){$x$};
        \path[-{Latex[round]}]
        	\edo(v2,v1,2,below)
        	\edo(v2,v3,1,left)
            \edo(v4,v3,2,above)
            \edo(v4,v5,1,above)
            \edo(v6,v5,2,right)
            \edo(v6,v1,1,below);
        \path[-{Latex[round]}, red]
            \edo(v1,u,3,right)
            \edo(v5,u,3,above);
        \path[-{Latex[round]}, red, dashed]
            \edo(v3,u,3,above);
        \path[-{Latex[round]}]
            (x) edge[bend right] node[below, inner sep=2pt]{4} (v1)
            \edo(x,v2,4,below)
            \edo(x,v3,4,above)
            (x) edge[bend left] node[above, inner sep=2pt]{4} (v4)
            (x) edge[bend left=90] node[above, inner sep=2pt]{4} (v5)
            (x) edge[bend right=90] node[below, inner sep=2pt]{4} (v6);
    \end{tikzpicture}}
    \caption{This graph shows the starting configuration of a best response cycle. The agents $v_1$, $v_3$ and $v_5$ remove and add the (red) edges towards $u$ in a circular fashion.}
    \label{fig:improving_cycle}
\end{figure}
\begin{proof}
    It follows directly from the existence of best response cycles.
    \Cref{fig:improving_cycle} shows an example of a best response cycles.
    By activating nodes $v_1, v_3, v_5, v_1, v_3,v_5$ in that order, toggling edges $(v_1,u), (v_3,u), (v_5,u), (v_1,u), (v_3,u), (v_5,u)$ are improving moves for the respective agents resulting in the starting configuration.
\end{proof}

\theoremUBsocopt*
\begin{proof}
    It follows directly from the algorithm provided in~\cite{CasteigtsPS21} that computes a temporal spanner of $H$ with $\mathcal{O}(n\log(n))$ edges.
\end{proof}

\PoALifetime*
\begin{proof}
    The upper bound on $\poa_{\textnormal{GE}}(n,\maxLifetime)$ follows directly from~\Cref{lem:lifetime_stable} by lower bounding the social cost of the optimum with $n-1$.
    
    For the tight lower bound when $\maxLifetime=2$ which also holds for $\poa_{\textnormal{NE}}(n,2)$, we construct the following family of graphs. The host graph $H$ is defined as follows
    \begin{align*}
        V_H&\coloneqq\{v_1,v_2,\dots,v_n\}\\
        \lambda(e)&\coloneqq \begin{cases}
            1&\text{if }e=\{v_2,v_3\}\vee\exists i\ge 4\colon e=\{v_1,v_i\}\\
            2&\text{otherwise}.
        \end{cases}
    \end{align*}
    The graph $G \coloneqq G(\s)$ induced by the strategy profile $\s$ has the following set of edges
    \[
        E_G \coloneqq\big\{(v_1,v_2), (v_2,v_3)\big\}\cup\bigcup_{i=4}^n\big\{(v_3,v_i),(v_i,v_1)\big\}.
    \]
    See~\Cref{fig:lower_bound_constructions} for an illustration.
    
    Note that $G$ is a NE. Furthermore, it contains $2(n-2)$ edges, hitting the upper bound given in~\Cref{lem:lifetime_stable} for GEs. Since every NE is also a GE, this bound also applies to NEs. Furthermore, $H$ contains a spanning tree of label $2$ which is socially optimal and contains $n-1$ edges. Therefore, $\opt$ contains $n-1$ edges. As a consequence, we get
    \begin{equation*}
        \poa_{\textnormal{GE}}(n,2)=\poa_{\textnormal{NE}}(n,2)=\frac{2(n-2)}{(n-1)}=2-\frac{2}{n-1}.\qedhere
    \end{equation*}
\end{proof}

\begin{lemma}\label{lem:large_node_exists}
    Let $G$ be a directed temporal graph with $n$ nodes and at least $\sqrt{6}n^\frac{3}{2}+n$ edges. Then there is a node $z\in V_G$ and a set $M\subset V_G$ such that
    \begin{compactenum}
        \item $|M| = \lceil\frac{1}{3}\sqrt{6n}\rceil$;
        \item $(u,z) \in E_G$, for every $u \in M$;
        \item each $u\in M$ has at least $\frac{2}{3}\sqrt{6n}$ outgoing edges $e=(u,v)\in E_G$ with $z\neq v$ and $\lambda_G((u,z))\leq\lambda_G(e)$.
    \end{compactenum}
\end{lemma}
\begin{proof}
    Consider the graph $G'$ where for each node $v$ we remove the $\lceil\frac{2}{3}\sqrt{6n}\rceil$ outgoing edges with the largest labels (break ties arbitrarily).
    If $v$ has less outgoing edges, we just remove all of them. Now, $G'$ has at least $\sqrt{6}n^\frac{3}{2}+n-n\lceil\frac{2}{3}\sqrt{6n}\rceil\ge n \frac{1}{3}\sqrt{6n}$ edges. By the pigeonhole principle, there is a node $z$ with at least $\lceil\frac{1}{3}\sqrt{6n}\rceil$ incoming edges. Let $M$ be a set of $\lceil\frac{1}{3}\sqrt{6n}\rceil$ neighbors $u$ of $z$ in $G'$ that have a directed edge $e'=(u,z)\in E_{G'}$ towards $z$. By construction of $G'$, each $u\in M$ has at least $\frac{2}{3}\sqrt{6n}$ outgoing edges $e=(u,v)\in E_G$ such that $v \neq z$ and $\lambda_G(e')\leq\lambda_G(e)$  (see~\Cref{fig:graph_including_forbidden_structure}).
\end{proof}

\fi

\end{document}